\documentclass[11pt]{article}
\usepackage{url}
\usepackage{amsthm,epsfig,bm,booktabs,amssymb}
\usepackage{amsfonts}
\usepackage[listformat=empty]{caption}
\usepackage{amsmath,latexsym,amssymb,epsf,amsfonts}
\usepackage{epsf}
\usepackage{color}
\usepackage{mathrsfs}
\usepackage{epstopdf}
\usepackage{threeparttable}
\usepackage{graphicx}
\usepackage{subfig}
\usepackage{tabularx}
\usepackage[all]{xy}
\usepackage{float}
\usepackage{cases}
\usepackage{algorithm,algorithmic}

\newcommand{\ba}{\begin{align*}}
\newcommand{\ea}{\end{align*}}
\newcommand{\bi}{\begin{itemize}}
\newcommand{\ei}{\end{itemize}}

\newcommand{\hf}{\frac{1}{2}}

\newcommand{\bx} {\bar{x}}
\newtheorem{defin}{Definition}[section]

\newtheorem{thm}{Theorem}[section]
\newtheorem{prop}{Proposition}[section]

\newtheorem{rmk}{Remark}[section]

\newcommand{\lam}{\lambda}

\newcommand{\Ker}{{\rm Ker}}
\newcommand{\sgn}{{\rm sgn}}
\newcommand{\Rn}{\mathbb{R}^n}
\newcommand{\R}{\mathbb{R}}

\newcommand{\Rmn}{\mathbb{R}^{m\times n}}
\newcommand{\0}{\mathbf{0}}
\newcommand{\ones}{\mathbf{1}}

\newcommand{\xk}{x^{(k)}}
\newcommand{\T}{\mathrm{T}}
\begin{document}

\title{ITERATIVE $\ell_1$ MINIMIZATION FOR NON-CONVEX COMPRESSED SENSING}

\author{Penghang Yin
\thanks {Department of Mathematics, UC Irvine, Irvine,
CA 92697, USA, (penghany@uci.edu).}
\and Jack Xin \thanks {Department of Mathematics, UC Irvine, Irvine,
CA 92697, USA, (jxin@math.uci.edu).} }

\maketitle

\begin{abstract}
An algorithmic framework, based on the difference of convex functions algorithm (DCA), is proposed 
for minimizing a class of concave sparse metrics for compressed sensing problems. The resulting 
algorithm iterates a sequence of $\ell_1$ minimization problems.  
An exact sparse recovery theory is established to show that the proposed framework always 
improves on the basis pursuit ($\ell_1$ minimization) and inherits robustness from it. 
Numerical examples on success rates of sparse solution recovery illustrate further that, 
unlike most existing non-convex compressed sensing solvers 
in the literature, our method always out-performs basis pursuit, 
no matter how ill-conditioned the measurement matrix is. Moreover, 
the iterative $\ell_1$ (IL$_1$) algorithm lead by a wide margin 
the state-of-the-art algorithms on $\ell_{1/2}$ and logarithimic minimizations 
in the strongly coherent (highly ill-conditioned) regime, despite the same objective functions. Last but not least, in the application of magnetic resonance imaging (MRI), IL$_1$ algorithm easily recovers the phantom image with just 7 line projections.
\end{abstract}

\section{Introduction}
Compressed sensing (CS) techniques \cite{CT_05,CRT_06,candes2006stable,Don_06} 
enable efficient reconstruction of a sparse signal under linear measurements 
far less than its physical dimension.
Mathematically, CS aims to recover an $n$-dimensional vector $\bx\in\Rn$ with few non-zero components from an under-determined linear system $Ax = A\bx$ of just $m\ll n$ equations, where 
$A\in\Rmn$ is a known measurement matrix. The first CS technique is the convex $\ell_1$ minimization or the so-called basis pursuit \cite{DonHuo_01}:
\begin{equation}\label{L1}
\min_{x\in\Rn} \; \| x\|_1 \quad  \mbox{s.t.} \quad  Ax = A\bx.
\end{equation}
Breakthrough results \cite{CT_05} have established that 
when matrix $A$ satisfies certain restricted isometry property (RIP), 
the solution to (\ref{L1}) is exactly $\bx$. It was shown that 
with overwhelming probability, several random ensembles such as random Gaussian, random
Bernoulli, and random partial Fourier matrices, are of RIP type  \cite{CT_05, CDD, R}.   
Note that (\ref{L1}) is just a minimization principle rather than an algorithm for retrieving $\bx$. Algorithms for solving (\ref{L1}) and its associated $\ell_1$ regularization problem \cite{lasso}:
\begin{equation}\label{regL1}
\min_{x\in\Rn} \; \frac{1}{2}\|Ax - b\|^2 + \lam\|x\|_1
\end{equation}
include Bregman methods \cite{YOGD_08,splitBregman}, 
alternating direction algorithms \cite{Boyd_11,yall1,Ernie}, 
iterative thresholding methods \cite{Daub_04,FISTA} among others \cite{FPC}.

Inspired by the success of basis pursuit, researchers then 
began to investigate various non-convex CS models and algorithms. 
More and more empirical studies have shown that non-convex CS methods 
usually outperform basis pursuit when matrix $A$ is RIP-like, 
in the sense that they require fewer linear measurements to 
reconstruct signals of interest. Instead of minimizing $\ell_1$ norm, 
it is natural to consider minimization of non-convex (concave) sparse metrics, 
for instance, $\ell_q$ (quasi-)norm ($0<q<1$) \cite{chartrand08, CY, LXY}, 
capped-$\ell_1$ \cite{Zhang_08, LYX}, and transformed-$\ell_1$ \cite{transformed-l1,ZX}. 
Another category of CS methods in spirit rely on support detection of $\bx$. 
To name a few, there are orthogonal matching pursuit (OMP) \cite{OMP}, 
iterative hard thresholding (IHT) \cite{IHT}, (re)weighted-$\ell_1$ scheme \cite{rwl1}, 
iterative support detection (ISD) \cite{ISD}, and their variations \cite{CoSaMP,ZXWKQ, sorted}.

On the other hand, it has been proved that even if $A$ is not RIP-like and contains highly correlated columns, basis pursuit still enables sparse recovery under certain conditions of $\bx$ involving its support \cite{Super:candes2014}. In this scenario, most of the existing non-convex CS methods, however, are not that robust to the conditioning of $A$, as suggested by \cite{YLHX}. Their success rates will drop as columns of $A$ become more and more correlated. In \cite{YLHX}, based on the difference of convex functions 
algorithm (DCA) \cite{DCA:tao1997convex,DCA:tao1998dc}, the authors propose DCA-$\ell_{1-2}$ for minimizing the difference of $\ell_1$ and $\ell_2$ norms \cite{ELX, YEX}.
Extensive numerical experiments \cite{YLHX, LYHX, LYX} imply that DCA-$\ell_{1-2}$ algorithm consistently outperforms $\ell_1$ minimization, irrespective of the conditioning of $A$.

Stimulated by the empirical evidence found in \cite{YLHX, LYHX, LYX}, we propose a general DCA-based CS framework for the minimization of a class of concave sparse metrics. More precisely, 
we consider the reconstruction of a sparse vector $\bx\in\Rn$ by minimizing sparsity-promoting metrics: 
\begin{equation}\label{CS}
\min_{x\in\Rn}  P(|x|) \quad \mbox{s.t.} \quad Ax = A\bx.
\end{equation}
Throughout the paper, we assume that $P(x)$ always takes the form $\sum_{i=1}^n p(x_i)$ unless otherwise stated,  where $p$ defined on $[0,+\infty)$ satisfies:
\begin{itemize}
\item[$\bullet$] $p$ is concave and increasing. 
\item[$\bullet$] $p$ is continuous with the right derivative $p'(0+)>0$.
\end{itemize}
The first condition encourages zeros in $|x|$ rather than small entries, since $p$ changes rapidly around the origin; the second one is imposed for the good of the proposed algorithm, as will be seen later.
A number of sparse metrics in the literature enjoy the above properties, including smoothly clipped absolute
deviation (SCAD) \cite{SCAD}, capped-$\ell_1$, transformed-$\ell_1$, and of course $\ell_1$ itself. Although $\ell_q$ ($q\in(0,1)$) and logarithm functional do not meet the second condition, their smoothed versions $p(t) = (t+\varepsilon)^q$ and $p(t) = \log(t+\varepsilon)$ are differentiable at zero. These proposed properties will be essential in the algorithm design as well as in the proof of main results. 

Our proposed algorithm calls for solving a sequence of minimization subproblems. 
The objective of each subproblem is $\|x\|_1$ plus a linear term, which is convex and tractable. 
We further validate robustness of this framework, by showing theoretically and 
numerically that it performs at least as well as basis pursuit 
in terms of uniform sparse recovery, independent of the conditioning of $A$ and sparsity metric.  

The paper is organized as follows. In section 2, we overview RIP and coherence of sensing matrices, 
as well as descent property of DCA. In section  3, we provide the iterated $\ell_1$ framework for 
non-convex minimization, with worked out examples on representative sparse objectives including the total 
variation. In section 4, we prove the main exact recovery results based on unique recovery 
property of $\ell_1$ minimization instead of RIP, which forms a theoretical basis of 
the better performance of DCA. In section 5, we compare iterative $\ell_1$ algorithms 
with two state-of-the-art non-convex CS algorithms, IRLS-$\ell_q$ \cite{LXY} and IRL$_1$ \cite{rwl1}, 
and ADMM-$\ell_1$, in CS test problems with varying degree of coherence. We find that 
 iterative $\ell_1$ outperforms ADMM-$\ell_1$ independent of the sensing matrix coherence, and 
leads IRLS-$\ell_q$ \cite{LXY} and IRL$_1$ \cite{rwl1} in the highly coherent regime. This is 
consistent with earlier findings of DCA-$\ell_{1-2}$ algorithm \cite{YLHX, LYHX, LYX} to which 
our theory also applies. We also evaluate these two non-convex metrics on a two-dimensional example of reconstructing MRI from a small number of projections, our iterative $\ell_1$ algorithm succeed with 7 projections for both metrics. Using the same objective functions, the state-of-the-art algorithms need at least 10 projections.  Concluding remarks are in section 6.

\bigskip 

{\bf Notations.} Let us fix some notations. For any $x, y\in\Rn$, $\langle x,y \rangle = x^\T y$ is their inner product. 
$\0\in\Rn$ is the vector of zeros, and similar to $\ones$. $\circ$ is Hadamard (entry-wise) product, meaning that $x\circ y = \sum_i^n x_iy_i$. $I_m$ is the identity matrix of dimension $m$.
For any function $g$ on $\Rn$, $\nabla g(x)\in \partial g(x)$ is a subgradient of $g$ at $x$. The $\sgn(x)$ is the signum function on $\R^n$ defined as
\begin{equation*}
(\sgn(x))_i:=
\begin{cases}
\frac{x_i}{|x_i|} & \text{if} \; x_i\neq0,\\
0 & \text{if} \; x_i=0.
\end{cases}
\end{equation*}
For any set $\Omega\subseteq\Rn$, $\iota_{\Omega}(x)$ is given by
\begin{equation*}
\iota_{\Omega}(x):=
\begin{cases}
0 & \text{if} \; x\in\Omega,\\
\infty & \text{if} \; x\not\in\Omega.
\end{cases}
\end{equation*}

\section{Preliminaries}
The well-known CS concept during the past decade is the restricted isometry property (RIP) introduced by Cand\`{e}s \textit{et al.} \cite{CT_05}, which is used to characterize matrices that are nearly orthonormal.\begin{defin}
For each number $s$,
$s$-restricted isometry constant of $A$ is the smallest $\delta_s\in(0,1)$
such that
$$
(1-\delta_s)\|x\|_2^2 \leq \|A x\|_2^2 \leq (1+\delta_s)\|x\|_2^2
$$
for all $x\in\Rn$ with sparsity of $s$.
The matrix $A$ is said to satisfy the $s$-RIP with $\delta_s$.
\end{defin}
Mutual coherence \cite{DonHuo_01} is another commonly-used concept closely related to the success of CS task.

\begin{defin}
The coherence of a matrix $A$ is the maximum absolute value of the cross-correlations between the columns of $A$, namely
$$
\mu(A) := \max_{i\neq j} \dfrac{|A_i^\T A_j|}{\|A_i\|_2\|A_j\|_2}.
$$
\end{defin}

When matrix $A$ have small mutual coherence (incoherent) or small RIP constant, its columns tend to be more separated or distinguishable, which is intuitively favorable to identification of the supports of target signal. On the other hand, a highly coherent matrix with large coherence poses challenge to the reconstruction.

Next we give a brief review on the difference of convex functions algorithm (DCA). DCA has been widely applied to sparse optimization problems in several works \cite{YLHX, ELX, NoncvxDC, transformed-l1, LYX}.
For an objective function $F(x) = G(x) - H(x)$  on the space $\Rn$, where $G(x)$ and $H(x)$ are lower semicontinuous proper convex functions, we call $G-H$ a DC decomposition of $F$.

DCA takes the following form
\begin{equation*}\label{dcaiter}
\begin{cases}
y^{(k)} \in \partial H(\xk)\\
x^{(k+1)} = \arg\min_{x\in\Rn} G(x) - (H(\xk) + \langle y^k, x-x^k \rangle)
\end{cases}
\end{equation*}
Since $y^{(k)} \in \partial H(\xk)$, by the definition of subgradient, we have
$$
H(x^{k+1})\geq H(x^k) + \langle y^k, x^{k+1}-x^k \rangle.
$$
Consequently,
$$
G(\xk)-H(\xk)\geq G(x^{(k+1)}) - (H(\xk) + \langle y^{(k)}, x^{(k+1)}-\xk \rangle)\geq G(x^{(k+1)})- H(x^{(k+1)}).
$$
The fact that $x^{(k+1)}$ minimizes $G(x) - (H(\xk) + \langle y^k, x-\xk \rangle)$ was used in the first inequality above.
Therefore, DCA permits a decreasing sequence $\{F(\xk)\}$, leading to its convergence provided $F(x)$ is bounded from below.

\section{Iterative $\ell_1$ framework}
Our proposed iterative $\ell_1$ framework for solving (\ref{CS}) is built on $\ell_1$ minimization and DCA. Note that (\ref{CS}) can be equivalently written as 
$$
\min_{x\in\Rn}  P(|x|) + \iota_{\{x:Ax = A\bx\}}(x).
$$
We then rewrite the above objective in DC decomposition form:
$$
P(|x|) + \iota_{\{x:Ax = A\bx\}} (x)=  (p'(0+)\|x\|_1 + \iota_{\{x:Ax = A\bx\}} (x)) - (p'(0+)\|x\|_1 - \sum_{i=1}^np(|x_i|))
$$
Clearly the first term on the right-hand side is convex in terms of $x$. We show below that the second term is also a convex function. 

\begin{prop}\label{prop} 
$p'(0+)\|x\|_1 - \sum_{i=1}^np(|x_i|)$ is convex in $x$. 
\end{prop}

\begin{proof}
For notational convenience, define $f(t) := p'(0+)t-p(t)$ on $[0,\infty)$. Since p is concave on $[0,\infty)$, we have that $f$ is convex on $[0,\infty)$. We only need to show that $f(|\cdot|)$ is convex on $\mathbb{R}$, or equivalently, for all $t_1,t_2\in\R$, $a\in(0,1)$,
$$
f(|at_1+(1-a)t_2|)\leq af(|t_1|) + (1-a)f(|t_2|).
$$  

{\bf Case 1.} If $t_1$ and $t_2$ have the same sign or one of them is $0$. Since $f(|at_1+(1-a)t_2|) = f(a|t_1|+(1-a)|t_2|)$ and $f$ is convex on $[0,\infty)$, then the above inequality holds. 

{\bf Case 2.} If $t_1$ and $t_2$ are of the opposite sign.  By the concavity of $p$ on $[0,\infty)$, we have
$$
p(t)\leq p(0) + p'(0+)t, \; \forall t>0,
$$
that is, $f(t)\geq f(0)$ for all $t>0$. Without loss of generality, we suppose $a|t_1|\geq (1-a)|t_2|$. Then
\begin{align*}
& f(|at_1+(1-a)t_2|) = f(a|t_1|-(1-a)|t_2|)\\
\leq & \frac{(1-a)(|t_1|+|t_2|)}{|t_1|}f(0)+\frac{a|t_1|-(1-a)|t_2|}{|t_1|}f(|t_1|)\\
\leq & (1-a)f(|t_2|) + \frac{(1-a)|t_2|}{|t_1|}f(|t_1|) +\frac{a|t_1|-(1-a)|t_2|}{|t_1|}f(|t_1|)\\
= & af(|t_1|) + (1-a)f(|t_2|)
\end{align*}
In the first inequality above, we used the convexity of $f$ on $[0,\infty)$, whereas in the second one, we used the fact that $f(t)\geq f(0)$ for $t>0$.
\end{proof}

At the $(k+1)^{\mathrm{th}}$ iteration, DCA calls for linearization of the second convex term at the current guess $x^{(k)}$, and solving the resulting convex subproblem for $x^{(k+1)}$. After converting back the linear constraint and removing the constant and the factor of $p'(0+)$, we iterate:
\begin{equation}\label{DCA}
x^{(k+1)} = \arg\min_{x} \|x\|_1 - \langle R(\xk),x\rangle \quad \mbox{s.t.} \quad Ax = A\bx,
\end{equation}
where 
$$
R(x):= \sgn(x)\circ(\ones- \frac{P'(|x|)}{p'(0+)})\in\partial(\|\cdot\|_1-\frac{P(|\cdot|)}{p'(0+)})(x).
$$ 
Be aware that $P'(|x|)\in\partial P(\cdot)(|x|)$ denotes subgradient of $P$ at $|x|$ rather than subgradient of $P(|\cdot|)$ at $x$. In this way, the subproblem reduces to minimizing $\|x\|_1$ plus a linear term of $x$, which can be effciently solved by a variety of state-of-the-art algorithms for basis pursuit (with minor modifications).
In Table \ref{table}, we list some non-convex metrics and the corresponding iterative $\ell_1$ algorithm. 

\begin{table}[htbp]
\caption{Examples of sparse metrics and associated iterative $\ell_1$ scheme}\label{table}
\centering
\begin{tabular}{p{80pt}p{100pt}p{60pt}p{120pt}}
\toprule
sparse metric & $p(t)$ & $p'(0+)$ & $(R(x))_i$\\
\midrule
capped-$\ell_1$ & $\min\{t,\theta\},\;\theta>0$ & 1 & $\sgn(x_i)\iota_{|x_i|\geq \theta}$\\
transformed-$\ell_1$ & $\frac{(\theta+1)t}{t+\theta},\;\theta>0$ & $\frac{\theta+1}{\theta}$ & $\sgn(x_i)(1-(\frac{\theta}{|x_i|+\theta})^2)$\\
smoothed log & $\log(t+\varepsilon),\;\varepsilon>0$ & $\frac{1}{\varepsilon}$ & $\sgn(x_i)(1-\frac{\varepsilon}{|x_i|+\varepsilon})$\\
smoothed $\ell_q$ & $(t+\varepsilon)^q,\;\varepsilon>0$ & $q\varepsilon^{q-1}$ & $\sgn(x_i)(1-(\frac{\varepsilon}{|x_i|+\varepsilon})^{1-q})$\\
\bottomrule
\end{tabular}
\end{table}
For initialization, we take $x^{(0)}=R(x^{(0)})=0$, which is basically $\ell_1$ minimization. The proposed algorithm is thus summarized in Algorithm \ref{alg1} below. Due to the descending property of DCA, Algorithm \ref{alg1} produces a convergent sequence $\{P(\xk)\}$. Beyond that, we shall not prove any stronger convergence result on the iterates $\{x^k\}$ itself in this paper. The reason is that the convergence analysis may vary individually by choice of sparse metric. We refer the readers to \cite{YLHX} and \cite{ZX}, in which subsequential convergence of $\{\xk\}$ is established for DCA-$\ell_{1-2}$ and DCA-transformed-$\ell_1$ respectively.

\begin{algorithm}
\caption{Iterative $\ell_1$ minimization}
Initialize: $x^{(0)} = \0$.
\begin{algorithmic}\label{alg1}
\FOR {$k = 1,2,\dots$}
\STATE $y^{(k)} = \sgn(\xk)\circ(\ones- \frac{P'(|\xk|)}{p'(0+)})$
\STATE $x^{(k+1)} = \arg\min_{x} \|x\|_1 - \langle y^{(k)},x\rangle \quad \mbox{s.t.} \quad Ax = A\bx$
\ENDFOR
\end{algorithmic}
\end{algorithm}

{\bf Extensions.} Two natural extensions of (\ref{CS}) are regularized model:
\begin{equation}\label{reg}
\min_{x\in\Rn}  \hf\|Ax - b\|_2^2 + \lambda P(|Dx|),
\end{equation}
and denoising model:
\begin{equation}\label{denoise}
\min_{x\in\Rn}  P(|Dx|) \quad \mbox{s.t.} \quad \|Ax - b\|_2\leq \sigma,
\end{equation}
where $b$ is the measurement, $D$ is a general matrix, and $\lambda, \sigma>0$ are parameters. They find applications in magnetic resonance imaging \cite{CRT_06}, total variation denoising \cite{TV} and so on. We can show that DC decomposition of $P(|Dx|)$ is 
\begin{equation}\label{DC}
P(|Dx|) = p'(0+)\|Dx\|_1 - (p'(0+)\|Dx\|_1 - P(|Dx|)).
\end{equation}
The iterative $\ell_1$ frameworks are detailed in Algorithms \ref{alg2} and \ref{alg3} respectively. 

\begin{algorithm}
\caption{Iterative $\ell_1$ regularization}
Initialize: $x^{(0)} = \0$.
\begin{algorithmic}\label{alg2}
\FOR {$k = 1,2,\dots$}
\STATE $y^{(k)} = D^{\T}\left(\sgn(D\xk)\circ(\ones- \frac{P'(|D\xk|)}{p'(0+)})\right)$
\STATE $x^{(k+1)} = \arg\min_{x} \hf\|Ax-b\|_2^2 + \lambda p'(0+) (\|Dx\|_1 - \langle y^{(k)},x\rangle)$ 
\ENDFOR
\end{algorithmic}
\end{algorithm}

\begin{algorithm}
\caption{Iterative $\ell_1$ denoising}
Initialize: $x^{(0)} = \0$.
\begin{algorithmic}\label{alg3}
\FOR {$k = 1,2,\dots$}
\STATE $y^{(k)} = D^{\T}\left(\sgn(D\xk)\circ(\ones- \frac{P'(|D\xk|)}{p'(0+)})\right)$
\STATE $x^{(k+1)} = \|Dx\|_1- \langle y^{(k)},x\rangle \quad \mbox{s.t.} \quad \|Ax - b\|_2\leq \sigma$
\ENDFOR
\end{algorithmic}
\end{algorithm}

\section{Recovery results}
Although in general global minimum is not guaranteed in minimization, we can show that its performance is provably robust to the conditioning of measurement matrix $A$, by proving that it always tends to sharpen $\ell_1$ solution.

Let us take another look at the assumptions on $p$ which were crucial in the proof of Proposition \ref{prop}. Since $p$ is concave and increasing on $[0,\infty)$, we have 
$$
0\leq(\frac{P'(|x|)}{p'(0+)})_i\leq1, \quad \forall x\in\R, 1\leq i\leq n,
$$ 
and thus $\|R(x)\|_\infty\leq1$. Now we are ready to show the main results.

\begin{thm}[Support-wise uniform recovery]\label{uniform}
Let $T\subseteq\{1,\dots,n\}$ be an arbitrary but fixed index set. 
If basis pursuit uniquely recovers all $\bx$ supported on $T$, so does (\ref{DCA}). 
\end{thm}
\begin{proof}
By the assumption that basis pursuit uniquely recovers all $\bx$ supported on $T$, and by the 
well-known null space property \cite{FR_13} for $\ell_1$ minimization, we must have 
$$
\|h_T\|_1<\|h_{T^c}\|_1, \; \forall h\in\Ker(A)\setminus\{\0\},
$$
and $x^{(1)} = \bx$ in (\ref{DCA}). The $2^{\mathrm{nd}}$ step of DCA reads
$$
x^{(2)} = \arg\min \|x\|_1 - \langle R(\bx),x \rangle \quad \mbox{s.t.} \quad Ax = A\bx.
$$ 
Let $x^{(2)} = \bx + h^{(2)}$, then 
\begin{align*}
&\|\bx\|_1 - \langle R(\bx),\bx \rangle \geq \|\bx +h^{(2)}\|_1 -\langle R(\bx),\bx+h^{(2)} \rangle \\
\Longrightarrow & \|\bx\|_1 - \langle R(\bx),\bx \rangle \geq \|\bx\|_1 + \langle \sgn(\bx), h_T^{(2)}\rangle +\|h^{(2)}_{T^c}\|_1 -\langle R(\bx),\bx+h^{(2)}_T \rangle \\
\Longleftrightarrow & -\langle \sgn(\bx) - R(\bx),h^{(2)}_T \rangle\geq \|h^{(2)}_{T^c}\|_1 \\
\Longleftrightarrow & -\langle \sgn(\bx)\circ\frac{P'(|\bx|)}{p'(0+)},h^{(2)}_T \rangle\geq \|h^{(2)}_{T^c}\|_1
\end{align*}
Since $\|\frac{P'(|\bx|)}{p'(0+)}\|_\infty \leq 1$, we have $\|h^{(2)}_T\|_1\geq\|h^{(2)}_{T^c}\|_1$. As a result, $h^{(2)}$ must be $\0$.
\end{proof}

If nonzero entries of $\bx$ have the same magnitude, a stronger result holds that (\ref{DCA}) recovers any fixed signal whenever basis pursuit does.
\begin{thm}[Recovery of equal-height signals]\label{single} Let $\bx$ be a 
signal with equal-height peaks supported on $T$, i.e. $|x_i| = |x_j|, \forall i,j\in T$. 
If the basis pursuit uniquely recovers $\bx$, so does (\ref{DCA}).
\end{thm}
\begin{proof}
If basis pursuit uniquely recovers $\bx$, then for all $h\in\Ker(A)\setminus\{\0\}$,  $\|\bx\|_1<\|\bx+h\|_1=\|\bx+h_T\|_1 + \|h_{T^c}\|_1$. This implies that for all $h\in\Ker(A)\setminus\{\0\}$ and $\|h\|_\infty \leq \min_{i\in T}|\bx_i|$,  $\|\bx\|_1<\|\bx+h_T\|_1 + \|h_{T^c}\|_1 = \|\bx\|_1 + \langle \sgn(\bx), h_T\rangle +\|h_{T^c}\|_1$.
So for all $h\in\Ker(A)\setminus\{\0\}$ and $\|h\|_\infty \leq \min_{i\in T}|\bx_i|$, we have $-\langle \sgn(\bx), h_T\rangle<\|h_{T^c}\|_1$.

Therefore, 
\begin{equation}\label{nsp}
-\langle \sgn(\bx), h_T\rangle<\|h_{T^c}\|_1, \; \forall h\in\Ker(A)\setminus\{\0\},
\end{equation}
and also $x^{(1)} = \bx$. 

We let $x^{(2)} = \bx + h^{(2)}$, and suppose that $h^{(2)}\neq\0$. Repeating the argument in Theorem \ref{uniform} and by (\ref{nsp}), we arrive at
$$
-\langle \sgn(\bx)\circ\frac{P'(|\bx|)}{p'(0+)},h^{(2)}_T \rangle\geq \|h^{(2)}_{T^c}\|_1 > -\langle\sgn(\bx), h_T^{(2)}\rangle.
$$
Since peaks of $\bx$ have equal height, $(\frac{P'(|\bx|)}{p'(0+)})_i\in[0,1)$ is a constant for all $i\in T$. So $-\langle \sgn(\bx)\circ\frac{P'(|\bx|)}{p'(0+)},h^{(2)}_T \rangle$ is non-negative and less than $-\langle \sgn(\bx), h^{(2)}_T\rangle$, which leads to a contradiction.
\end{proof}

\begin{rmk}
Although the conditions proposed in section 1 are not applicable to the metric $\ell_{1-2}$ since it is not separable, it is not hard to generalize these conditions and iterative $\ell_1$ algorithm to accommodate this case. The resulting algorithm is exactly DCA-$\ell_{1-2}$ in \cite{YLHX}. We can also readily extend the recovery theory to DCA-$\ell_{1-2}$, with $P(x) = \|x\|_1-\|x\|_2$ and
\begin{equation*}
R(x)=
\begin{cases}
\frac{x}{\|x\|_2} &  \text{if} \; x\neq\0,\\
\0 & \text{if} \; x=\0.
\end{cases}
\end{equation*}
Theorem \ref{uniform} provides a theoretical explanation for the experimental observations made in \cite{YLHX,LYHX,LYX} that DCA-$\ell_{1-2}$ performs consistently better than $\ell_1$ minimization.
\end{rmk}

\section{Numerical experiments}
\subsection{Exact recovery of sparse vectors}
We reconstruct sparse vector $\bx$ using iterative $\ell_1$ algorithm (Algorithm \ref{alg2} with $D = I_n$) for minimizing the regularized model (\ref{reg}) with smoothed $\ell_q$ norm (IL$_1$-$\ell_q$) and smoothed logarithm functional (IL$_1$-log), and compare them with two state-of-the-art non-convex CS algorithms, namely IRLS-$\ell_q$ \cite{LXY} and IRL$_1$ \cite{rwl1}. Note that IRLS-$\ell_q$ and IRL$_1$ attempt to minimize $\ell_q$ and logarithm, respectively, and both involve a smoothing strategy in minimization. So it would be particularly interesting to compare IL$_1$-$\ell_q$ with IRLS-$\ell_q$, and IL$_1$-log with IRL$_1$. $q=0.5$ is chosen for IRLS-$\ell_q$ and IL$_1$-$\ell_q$ in all experiments. We shall also include ADMM-$\ell_1$ \cite{Boyd_11} for solving $\ell_1$ regularization (LASSO) in comparison.  

Experiments are carried out as follows. We first sample a sensing matrix $A\in\Rmn$, and generate a test signal $\bx\in\Rn$ of sparsity $s$ supported on a random index set with i.i.d. Gaussian entries. We then compute the measurement $A\bx$ and apply each solver to produce a reconstruction $x^*$ of $\bx$. The reconstruction is called a success if 
$$\frac{\|x^*-\bx\|_2}{\|\bx\|_2}<10^{-3}.$$
We run 100 independent realizations and record the corresponding success rates at different sparsity levels.

{\bf Matrix for test.} We test on random Gaussian matrix whose columns satisfy
$$
A_i \overset{{\rm i.i.d.}}{\sim} \mathcal{N}(\0,I_m/m), \quad i = 1,\cdots,n
$$
Gaussian matrices are RIP-like and have uncorrelated (incoherent) columns. For Gaussian matrix, we choose $m=64$ and $n=256$.

We also use more ill-conditioned sensing matrix of significantly higher coherence. Specifically, a randomly oversampled partial DCT matrix $A$ is defined as
$$
A_i = \frac{1}{\sqrt{m}}\cos(2i\pi\xi/F), \quad i = 1,\cdots,n
$$
where $\xi\in \mathbb{R}^m \sim \mathcal{U}([0,1]^m)$ whose components are uniformly and independently sampled from [0,1]. $F\in \mathbb{N}$ is the refinement factor. Coherence $\mu(A)$ goes up as $F$ increases. In this setting, it is still possible to recover the sparse vector $\bx$ if its spikes are sufficiently separated. Specifically, we randomly select a $T$ (support of $\bx$) so that
$$
\min_{j,k\in T}|j-k| \geq L,
$$
where $L$ is called the minimum separation. It is necessary for $L$ to be at least 1 Rayleigh length (RL) which is unity in the frequency domain \cite{Fann_12, Don_92}. In our case, the value of 1 RL equals $F$.
The testing matrix $A\in\mathbb{R}^{100\times1500}$, i.e. $m= 100$, $n = 1500$. We test at three coherence levels with $F = 5, 10, 15$. Note that $\mu(A)\approx 0.95$ for $F= 5$, $\mu(A)\approx 0.998$ for $F= 10$, and $\mu(A)\approx 0.9996$ for $F= 15$. We also set $L = 2F$ in experiments.

{\bf Algorithm implementation.}
For ADMM-$\ell_1$, we let $\lam = 10^{-6}$, $\beta = 1$ , $\rho = 10^{-5}$, $\epsilon^{\mathrm{abs}}=10^{-7}$, $\epsilon^{\mathrm{rel}}=10^{-5}$, and the maximum number of iterations \verb"maxiter" = 5000 \cite{Boyd_11,YLHX}. For IRLS-$\ell_q$, \verb"maxiter" = 1000, \verb"tol"$=10^{-8}$. For rewighted $\ell_1$, the smoothing parameter $\varepsilon$ is adaptively updated as introduced in \cite{rwl1}, and the outer iteration criterion is stopped if the relative error between two consecutive iterates is less than $10^{-2}$.
The weighted $\ell_1$ minimization subproblems is solved by the YALL1 solver (available at http://yall1.blogs.rice.edu/). The tolerance for YALL1 was set to $10^{-6}$. All other settings of the algorithms are set to default ones.

For IL$_1$-$\ell_q$, we let $\lambda = 10^{-6}$, and the smoothing parameter
$\varepsilon = \max\{\frac{|x^{(1)}|_{(d)}}{3}, 0.01\}$,
where $x^{(1)}$ is the output from the first iteration, which is also the solution to LASSO. $|x|_{(d)}$ denotes the $d^{\mathrm{th}}$ largest entry of $|x|$. We set $d$ to $\lfloor\frac{m}{4}\rfloor$.
For IL$_1$-log, $\varepsilon = \max\{|x^{(1)}|_{(d)}, 0.01\}$.
The subproblems are solved by alternating direction method of multipliers (ADMM), which is detailed in \cite{YLHX}. The parameters for solving subproblems are the same as that for ADMM-$\ell_1$.
 
{\bf Interpretation of results.} The plot of success rates is shown in Figure \ref{fig}. 
When $A$ is Guassian, we see that all non-convex CS solvers are comparable 
and much better than ADMM-$\ell_1$, with IRLS-$\ell_q$ being the best. 
For oversampled DCT matrices, we see that the success rates of IRLS-$\ell_q$ and IRL$_1$ drop 
as $F$ increases, whereas the proposed IL$_1$-$\ell_q$ and IL$_1$-log are robust and 
consistently outperform ADMM-$\ell_1$. 
\begin{figure}
\begin{center}
\begin{tabular}{cc}
Gaussian & DCT, F = 5  \\
\includegraphics[width=0.49\textwidth]{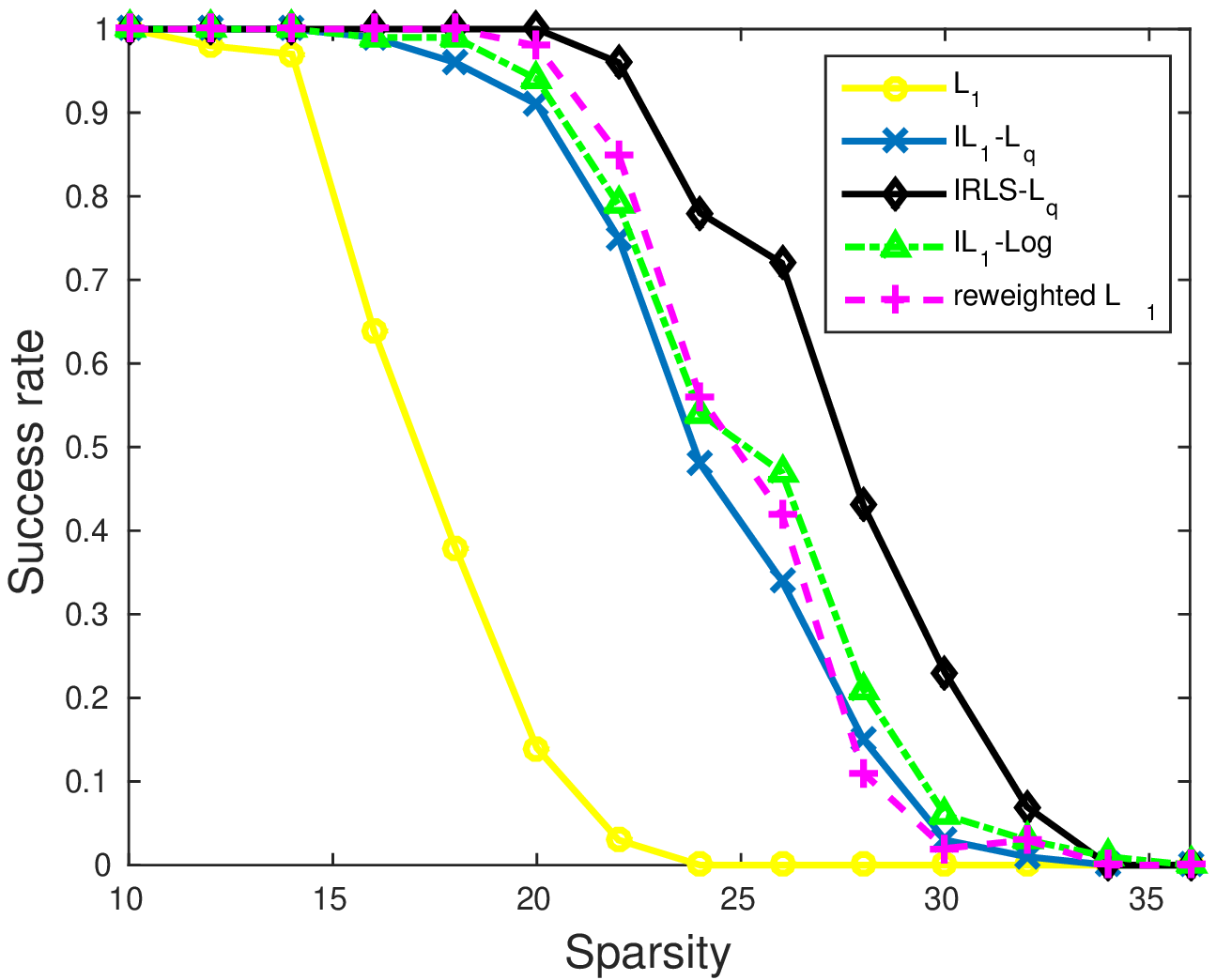}&
\includegraphics[width=0.49\textwidth]{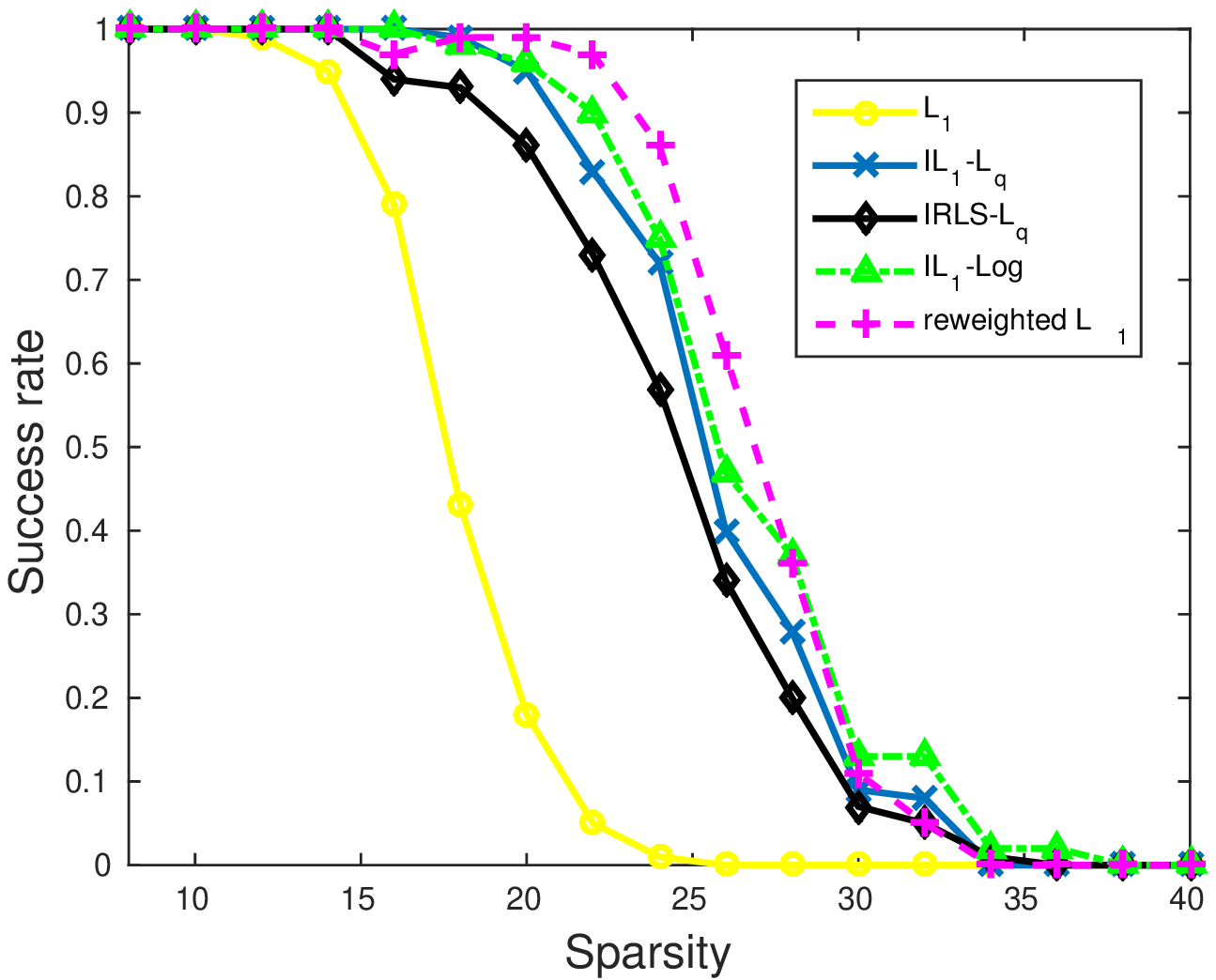}\\
DCT, F = 10 & DCT, F = 15 \\
\includegraphics[width=0.49\textwidth]{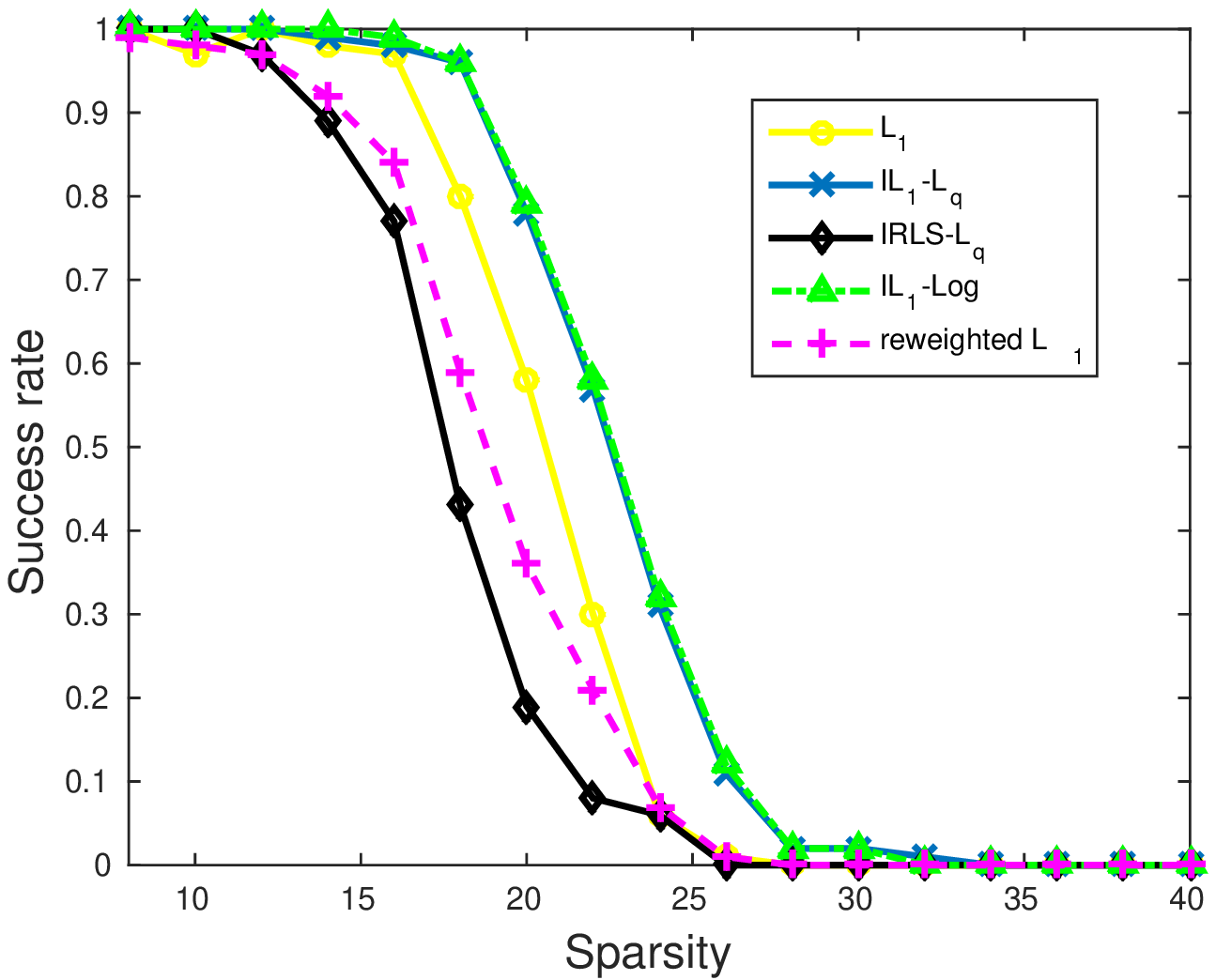}&
\includegraphics[width=0.49\textwidth]{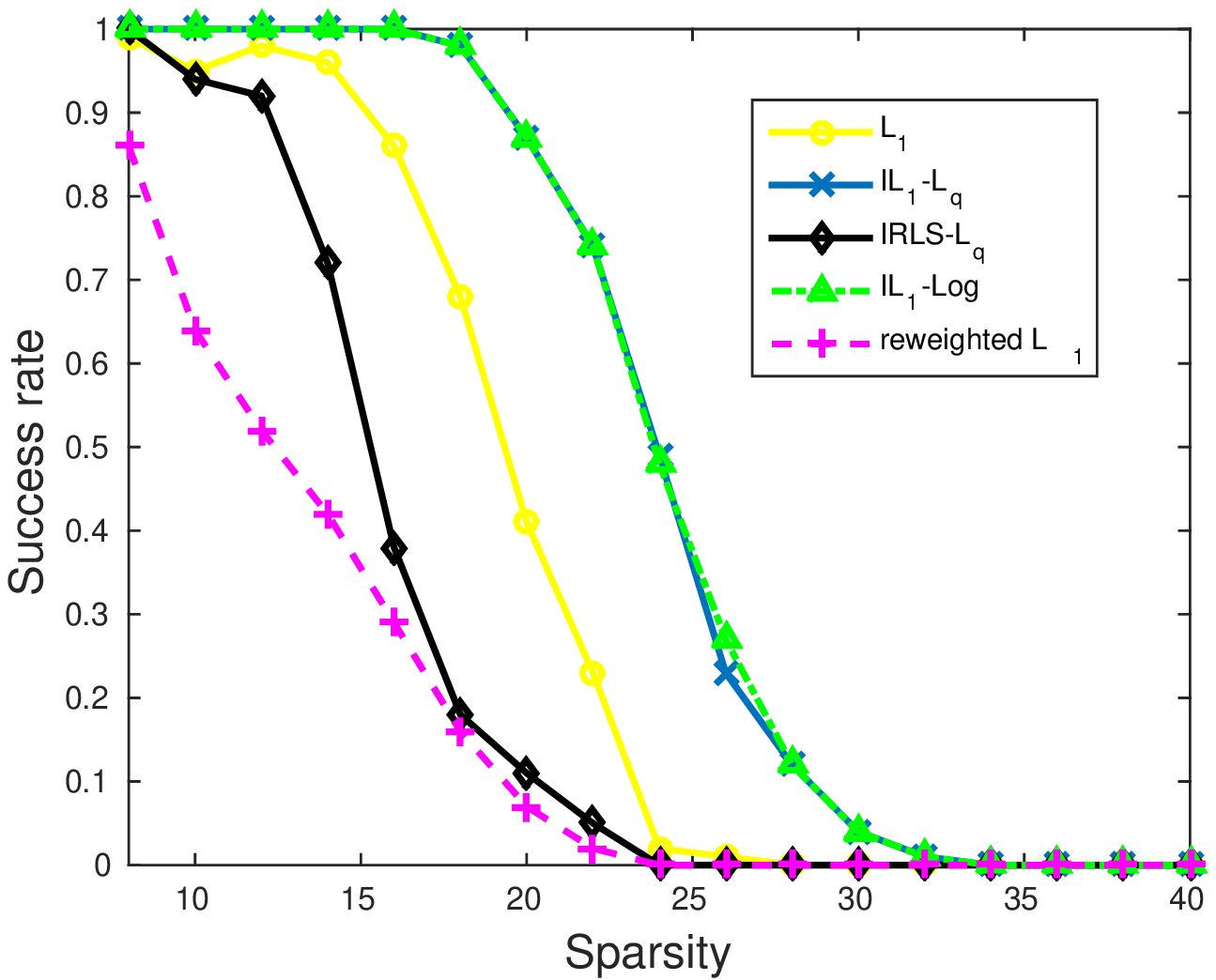}
 \\
\end{tabular}
\caption{Plots of success rates for comparing the iterative $\ell_1$ with other CS algorithms under 
the increasing coherence of the sensing matrices.} \label{fig}
\end{center}
\end{figure}

\subsection{MRI reconstruction}
We present an example of reconstructing the shepp-Logan phantom image of size $256\times 256$, to further demonstrate effectiveness of IL$_1$ algorithm. In this application, the sparsity of the gradient of the image/signal denoted by $u$ is exploited, which leads to the following minimization problem:
$$
\min_u P(|\nabla u|) \quad \mbox{s.t.} \quad S\mathcal{F}u = b,
$$
where $S$ denotes the sampling mask in the frequency domain, and $\mathcal{F}$ is the Fourier transform and $b$ the acquired data.
With $P(|\cdot|)$ being the $\ell_1$ norm, the above formulation reduces to the celebrated total variation (TV) minimization:
$$
\min_{u} \|\nabla u\|_1 \quad \mbox{s.t.} \quad S\mathcal{F}u = b.
$$
The above unconstrained problem together with its regularized problem 
\begin{equation}\label{reg_TV}
\min_{u} \hf\|S\mathcal{F}u -b\|_2^2 + \lambda \|\nabla u\|_1,
\end{equation}
can be sovled efficiently by split Bregman method \cite{splitBregman}, known to be equivalent to ADMM \cite{Ernie}. For general sparse metric $P$ (or $p$), (\ref{DC}) gives the DC decomposition
$$
P(|\nabla u|) = p'(0+)\|\nabla u\|_1 - (p'(0+)\|\nabla u\|_1 - P(|\nabla u|)),
$$
and thus its IL$_1$ algorithm for solving the regularized model takes the following form
\begin{equation*}
\begin{cases}
w^{(k)} =  \nabla^{\T}\left(\sgn(\nabla u^{(k)})\circ(\ones- \frac{P'(|\nabla u^{(k)}|)}{p'(0+)})\right)\\
u^{(k+1)} = \arg\min_{u} \hf\|S\mathcal{F}u -b\|_2^2 + \lambda p'(0+) (\|\nabla u\|_1 - \langle w^{(k)},u\rangle). \\
\end{cases}
\end{equation*}
Likewise the subproblem for updating $u^{(k+1)}$ above can also be solved by split Bregman, as the objective only differs by a linear term compared with (\ref{reg_TV}).

{\bf Numerical results.} In the experiment, we again choose $p$ to be the smoothed $\ell_q$ ($q = \hf$) and smoothed log respectively for the IL$_1$ algorithm, and set smoothing parameter $\varepsilon = 0.1$ and regularization parameter $\lambda = 10^{-6}$ for both implementations. The reconstruction results are shown in Figure \ref{fig2}. We find that 7 sampled projections are sufficient for both of the two penalties to recover the phantom image perfectly, in comparison to $\ell_1$ (TV) minimization which needs 10 projections for perfect image reconstruction by split Bregman. To the best of our knowledge, the other existing non-convex solvers for minimizing either $\ell_q$ or log penalties did no better than 10 projections \cite{rwl1, chartrand_09, chartrand_07}.  

\begin{figure}
\begin{center}
\begin{tabular}{cc}
7 sampled projections & $\ell_1$, rel. error =  0.4832 \\
\includegraphics[width=0.49\textwidth]{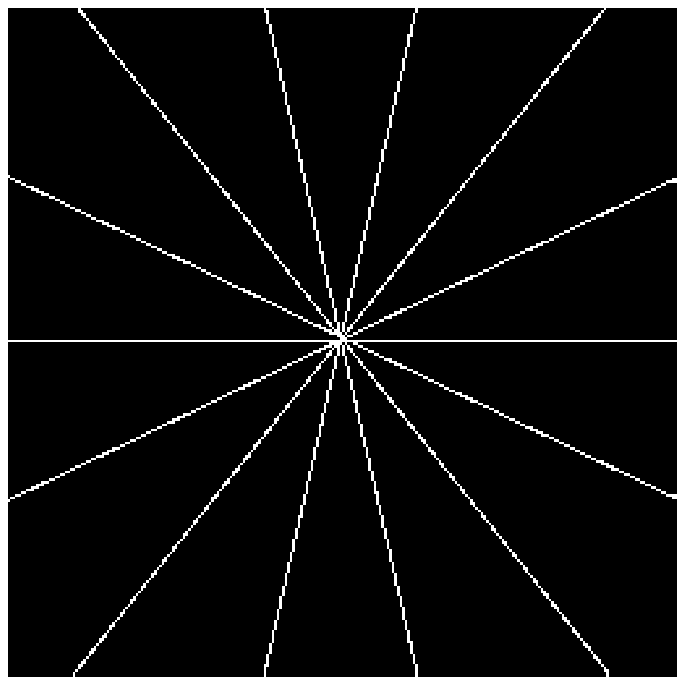}&
\includegraphics[width=0.49\textwidth]{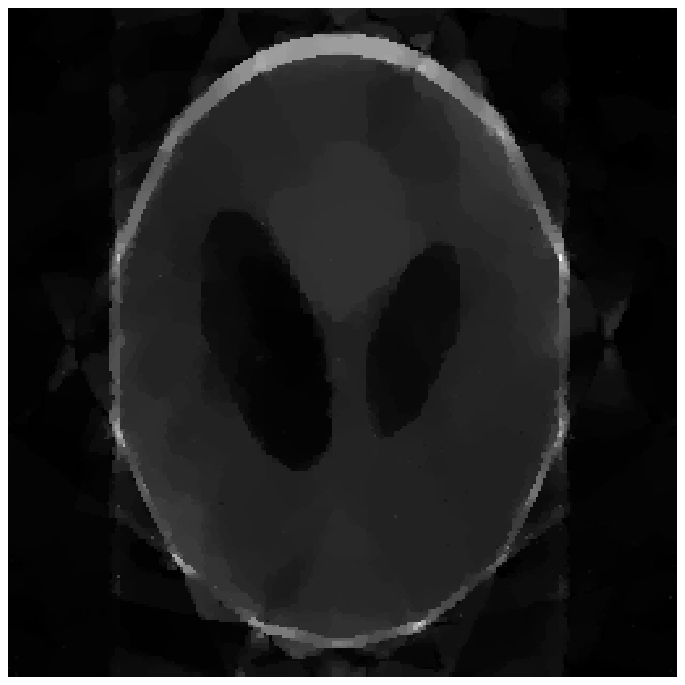}\\
IL$_1$-$\ell_q$, rel. error = $3.6\times10^{-9}$ & IL$_1$-log, rel. error = $4.8\times10^{-9}$ \\
\includegraphics[width=0.49\textwidth]{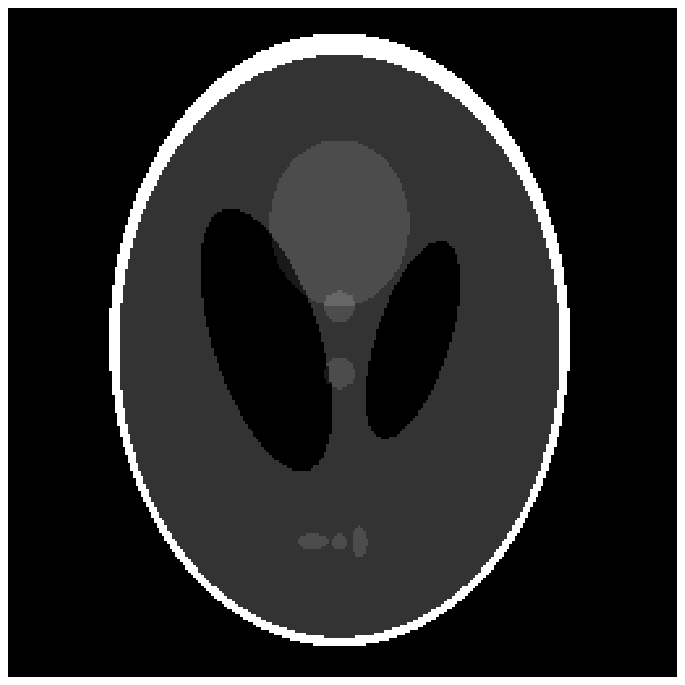}&
\includegraphics[width=0.49\textwidth]{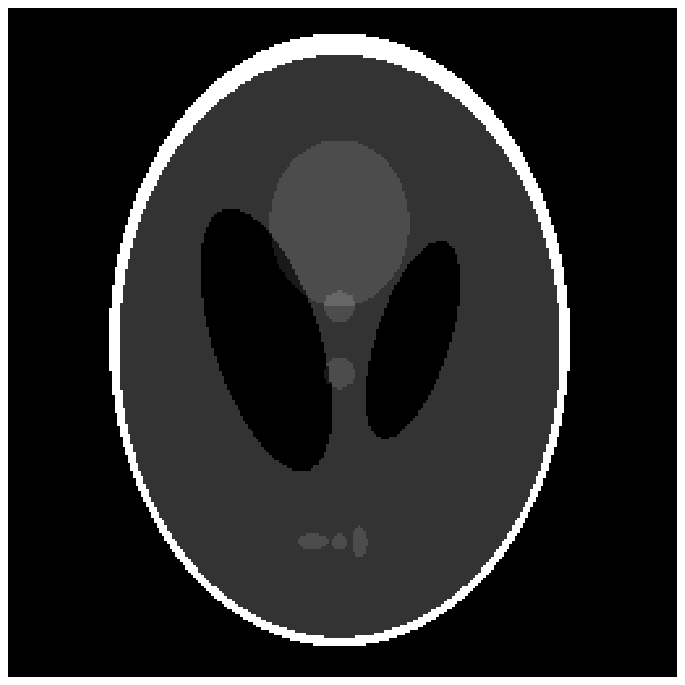}
 \\
\end{tabular}
\caption{Sampled lines and reconstructions for Shepp-Logan phatom image.} \label{fig2}
\end{center}
\end{figure}

\section{Conclusions}
We developed an iterative $\ell_1$ framework for a broad class of Lipschitz continuous 
non-convex sparsity promoting objectives, including those arising in statistics. 
The iterative $\ell_1$ algorithm is shown via theory and computation 
to improve on the $\ell_1$ minimization for CS problems independent 
of the coherence of the sensing matrices.

\bigskip

\noindent {\bf Acknowledgments.} The authors would like to thank 
Yifei Lou (University of Texas at Dallas) and Jong-Shi Pang (University of Southern California) 
for helpful discussions. The authors would also like to thank anonymous reviewers for their helpful comments.The work was partially supported by NSF grants DMS-1222507 and DMS-1522383.

\end{document}